\documentclass[11pt,a4paper,pdftex]{article}
\usepackage{amsmath, amscd, amssymb, amsthm, latexsym
}
\usepackage{amssymb}
\usepackage[colorlinks=true,linkcolor=red]{hyperref}
\usepackage{stackrel}
\usepackage{pdflscape}

\usepackage{pgf}
\usepackage{float}
\usepackage{xcolor}

\newtheorem{theorem}{Theorem}
\newtheorem{proposition}[theorem]{Proposition}
\newtheorem{corollary}[theorem]{Corollary}
\newtheorem{lemma}[theorem]{Lemma}

\theoremstyle{definition}
\newtheorem{definition}[theorem]{Definition}

\theoremstyle{remark}

\newtheorem{example}[theorem]{Example}

\newcommand{\N}{\mathbb{N}}

\newcommand{\nl}{\medskip\noindent}

\newcommand{\Prime}{\texttt{Prime}}

\newcommand{\LimPrime}{\texttt{LimPrime}}

\newcommand{\Div}{\texttt{Div}}
\newcommand{\Dom}{\texttt{Dom}}

\newcommand{\LCM}{\texttt{LCM}}
\newcommand{\lcm}{\text{lcm}}
\newcommand{\MULT}{\texttt{MULT}}

\newcommand{\TERM}{\texttt{TERM}}

\newcommand{\EQ}{\texttt{EQ}}

\newcommand \str[1] {$\langle  {#1} \rangle$}

\newcommand \strgen{$\langle \omega^{\omega^\lambda}; \times \rangle$}

\begin{document}
\title{ Complexity and (un)decidability of fragments of \strgen}

\author{Alexis B\`{e}s\\
	Universit\'e Paris-Est, LACL (EA 4219), UPEC, Cr\'eteil, France \\
	\and Christian Choffrut\\
	IRIF, CNRS and Universit\'e Paris 7 Denis Diderot, 
	France} 

\maketitle

\begin{abstract}
	We specify the frontier of decidability for fragments of the first-order theory of ordinal multiplication. We give a {NEXPTIME} lower bound for  the 
	complexity of the existential fragment of $\langle \omega^{\omega^\lambda}; \times, \omega, \omega+1, \omega^2+1 \rangle$ for every ordinal $\lambda$. Moreover, we prove (by reduction from Hilbert Tenth Problem) that 
	the $\exists^*\forall^{6}$-fragment  of $\langle \omega^{\omega^\lambda}; \times \rangle$ is undecidable  for every ordinal $\lambda$.
	
\end{abstract}

%

\section{Introduction}

The first-order  theory of ordinals was studied  under different signatures:  linear order,
addition, multiplication or some of their combinations. The purpose of this paper is to 
refine a  result  of the first author proving  that the first-order theory of the multiplicative 
monoid of an ordinal  is decidable if and only if 
it is less than $\omega^{\omega}$, \cite[Thm 11]{bes02}.  
We investigate the natural issue of trying to determine the boundary between decidability and undecidability for syntactic fragments of the theory. We first prove that the existential theory of \str{\N;+,|} (which was shown to be decidable by Lipshitz \cite{lipshitz}) is interpretable in the existential fragment of $\langle \omega^{\omega^\lambda}; \times, \omega, \omega+1, \omega^2+1 \rangle$ for every ordinal $\lambda$, where $\omega, \omega+1, \omega^2+1$ correspond to constants. By \cite{lechner}, this yields a {NEXPTIME} lower bound for  the 
complexity of the latter fragment. Then we prove, by reduction from Hilbert Tenth Problem, that 
 the $\exists^*\forall^{6}$-fragment  (an arbitrary number of existential variables followed 
 by $6$ universal variables) of $\langle \omega^{\omega^\lambda}; \times \rangle$ is undecidable  for every ordinal $\lambda$.
 Our results leave open the question of whether the existential theory of \str{\alpha;\times} is decidable for $\alpha \geq \omega^\omega$. Another related open question (raised in \cite{bes02}) is  whether one can decide satisfiability of systems of multiplicative equations over ordinals with constants. Both questions seem to be quite difficult. We give some insight about the difficulty in the final section of the paper.

\medskip We recall some historical background of  the first-order theory of ordinal arithmetic (see e.g. \cite{Rosenstein82}). The study of decidability and definability issues related to ordinal theories was initiated by Mostowski and Tarski who proved by means of quantifier-elimination that the class of well-ordered structures has a decidable elementary theory (\cite{MT49,DMT78}, see also \cite{Rab78}).
From the undecidability of first-order arithmetic (\cite{church}) one can deduce easily that for every ordinal $\xi$, the first-order theory of  \str{\omega^{\omega^{\xi}}; +, \times} is undecidable.
In the 1960's, by means of automata, B\"uchi (\cite{Buc65}, see also \cite{Mau97}) proved that for any ordinal $\alpha$, the weak monadic second-order theory of \str{\alpha; <} is decidable, from which he deduced decidability of the elementary theory of \str{2^\alpha;+}. This implies that the first-order theory of \str{{\omega^{\xi}}; +} is decidable for every ordinal $\xi$, since  $\omega^\xi=2^{\omega \xi}$. On the other hand, the second author proved in \cite{choffrutDLT2001} that the first-order theory of \str{{\omega^{\omega}}; +, x \mapsto \omega x} is undecidable.

Concerning the  decidability of  arithmetic without addition, i.e., of
\str{\omega; \times},  the result was announced by Skolem in \cite{Sko30}.  Mostowski proved it in \cite{Mos52} as a direct consequence  of  his results on direct products of structures and Presburger's decidability result for \str{\omega; +}. Other proofs can be found in \cite{Ceg81} and \cite{Hod82}. However, unlike the case of addition, decidability does not extend to all ordinals: in \cite{bes02} the first author proved that the theory of \str{\lambda; \times} is decidable if and only if $\lambda < \omega^\omega$. The result still holds if we replace $\times$ with the two binary predicates $|_r$ and $|_l$ where  $x|_ry$ (resp. $x|_ly$) means that $x$ is a right-hand (resp. left-hand) divisor of $y$ (both predicates are definable in \str{\lambda;\times}). 

\medskip We now give a brief outline of our paper. Section \ref{s:preliminaries}
contains the basic notions on ordinals such as the ordering and the two arithmetic operations
of addition and multiplication which should suffice even for the reader 
with no strong background in the theory of ordinals. 
In Section \ref{s:ordinals-as-a-monoid}  we show that 
three specific   constants which are  instrumental for our purpose can be expressed 
in \str{\lambda; \times} with formulas of low syntactic complexity. 
 We denote by $\Omega$
the structure \str{\lambda; \times} enriched with these constants 
and we show in Section \ref{s:interpretation-in-N-divibililty}
that the integers with the addition and the divisibility  can be 
interpreted in the existential fragment of $\Omega$
which gives us the above mentioned NEXPTIME  lower bound. 
In Section \ref{s:undecidability} we 
reduce the problem of 
solving Diophantine equations in the nonnegative integers to the $\exists^*\forall^{6}$-fragment
of \str{\lambda;\times} showing thus, 
via use Matijasevi\u{c} result that this fragment is undecidable.
In the last section \ref{s:final-observations} we observe that 
a simple  proof of the decidability of
the existential fragment of $\Omega$, would immediately 
offer an alternative proof of 
Makanin's result for  word equations.

\section{Preliminaries}
\label{s:preliminaries}

\subsection{Ordinal arithmetic}

We recall useful results about ordinal arithmetic. We refer the reader to the handbook 
of Sierpinski \cite{Sie65} for a more complete  exposition of the 
topic.

The following definition of the \emph{Cantor normal form}, abbreviated CNF,  is actually  a property.

\begin{definition}
\label{de:cantor-normal-form}
Every ordinal $\alpha>0$  has a unique form as a sum of 
decreasing $\omega$-powers with integer coefficients, namely
$$
\alpha = \omega^{\lambda_{r} } a_{r} +  \cdots + \omega^{\lambda_{1}} a_{1}, \quad 
$$
where $\lambda_{r} >  \cdots  >\lambda_{1}\geq 0$ are ordinals and
$a_{r},   \ldots, a_{1} >0$ are integers. The ordinal $\lambda_{r}$
is the \emph{degree} of $\alpha$, written $\partial(\alpha)$, and 
$\lambda_{1}$ its \emph{valuation} written $v(\alpha)$.
An ordinal is a \emph{successor} if $v(\alpha)=0$, otherwise it is a \emph{limit}. 

\end{definition}

\medskip We are given two ordinals in their normal forms
\begin{equation}
	\label{eq:cnf}
\alpha = \omega^{\lambda_{r} } a_{r} +  \cdots + \omega^{\lambda_{1}} a_{1}, \quad 
\beta = \omega^{\mu_{s} } b_{s} +  \cdots + \omega^{\mu_{1}} b_{1}, 
\end{equation}


%
The order is defined by  $\alpha<\beta$ if one of the following conditions is satisfied.

\nl 1) $r<s$ and for all $i=1, \ldots, r$ we have $a_{i}=b_{i}$ and $\lambda_{i}=\mu_{i}$.

\nl 2) for some $t<\min\{r,s\}$ and for all $i=1, \ldots, t$
we have $a_{i}=b_{i}$ and $\lambda_{i}=\mu_{i}$, 
and either $\lambda_{t+1}<\mu_{t+1}$ or $\lambda_{t+1}=\mu_{t+1}$
and $a_{t+1}<b_{t+1}$.

Furthermore, we define $0<\alpha$ for every nonzero ordinal $\alpha$.

\nl We now recall  the definition of the two arithmetic operations on the ordinals
by use of their Cantor normal form 

{
\begin{definition}
\label{de:sum}
The \emph{sum} $\alpha +\beta$ of $\alpha$ and $\beta$ given by their CNF as in (\ref{eq:cnf}) is 
$$
\begin{array}{ll}
\omega^{\lambda_{r} } a_{r} +  \cdots + \omega^{\lambda_{i}} a_{i} +
\omega^{\mu_{s} } b_{s} +  \cdots + \omega^{\mu_{1}} b_{1}  & \text{ if } \lambda_{i}> \mu_{s} >\lambda_{i-1},\\
\omega^{\lambda_{r} } a_{r} +  \cdots + \omega^{\lambda_{i+1}} a_{i+1} +
\omega^{\mu_{s} } (a_{i} + b_{r}) +  \cdots + \omega^{\mu_{1}} b_{1} & \text{ if }  \mu_{s} =\lambda_{i},\\
\beta  & \text{ if }  \mu_{s} >\lambda_{r}.
\end{array}
$$
\end{definition}

Furthermore, we have $0+\beta= \beta$ and $\alpha +0=\alpha$.
The sum is associative, has a neutral element $0$ and is not commutative.
It is left-cancellative ($\alpha+ \beta=\alpha+ \gamma$ implies $\beta= \gamma$)
but not right-cancellative.

\begin{definition}
\label{de:product}

 The \emph{product} $\alpha +\beta$  of $\alpha$ and $\beta$ given by their CNF as in (\ref{eq:cnf}) is 
\begin{equation}
\label{eq:product1}
\alpha \times \beta =  \omega^{\lambda_{r} + \mu_{s} } b_{s} +  \cdots + \omega^{\lambda_{r} +\mu_{1}} b_{1}
\end{equation}
if  $\mu_{1}>0$. 
If $\mu_{1}=0$  set  $\beta =\beta' + b_{1}$ where $0<b_{1}<\omega$ and $v(\beta')>0$.
Then
$$
\alpha \times \beta =  \alpha \times\beta' + \omega^{\lambda_{r} } (a_{r} \times b_1)   +  \omega^{\lambda_{r-1} } a_{r-1} + \cdots + \omega^{\lambda_{1}} a_{1},  
$$
yielding
\begin{multline}
\label{eq:product2}
\alpha \times \beta = \omega^{\lambda_{r}+\mu_{s}} b_{s} + \omega^{\lambda_{r}+\mu_{s-1}} b_{s-1} + 
 \cdots + \omega^{\lambda_{r}+ \mu_{2}} b_{2} + \\
\omega^{\lambda_{r} } a_{r}b_{1} +  \omega^{\lambda_{r-1} } a_{r-1} +\cdots + \omega^{\lambda_{1}} a_{1}.
\end{multline}
\end{definition}
}
Furthermore, we have $0\times \beta= \alpha\times 0=0$ for all $\alpha, \beta$.
The  multiplication is associative, has a neutral element $1$, is noncommutative,  
is left- (but not right-) \emph{cancellative} ($x\times y = x\times z\Rightarrow y=z$)
and left- (but not right-) distributes over the addition.

The next result is the unique combinatorial property 
of ordinals that we will need. 

\begin{lemma}{\cite[page 352]{Sie65}}
\label{le:commuting-successor-ordinals}
The product of two successor ordinals commutes if and only if they are of the 
form $\alpha^{j}$ and $\alpha^{n}$ for some ordinal $\alpha$ and some integers $j,n$. 
\end{lemma}

Observe  that this property  does not hold in general if the ordinals are limit, e.g.,
$(\omega^{2} + \omega) (\omega^{3} + \omega^{2})=(\omega^{3} + \omega^{2}) (\omega^{2} + \omega)$.

\subsection{Prime factorization}

We say that an ordinal $\alpha>0$ is  {\it prime} if it cannot be written as the
product of two ordinals less than  $\alpha$; an equivalent definition is that $\alpha$ admits exactly two  right-hand divisors. One proves  that there are three kinds of prime ordinals, \cite[page 336]{Sie65}: 
\emph{natural primes}  less than $\omega$ which are the usual primes, \emph{nonfinite successor primes} which are of the
 form $\omega^{\mu}+1$ for some $\mu>0$,
and \emph{limit primes} which are  of the form $\omega^{\omega^{\xi}}$ for some ordinal $\xi$.

The  factorization of natural 
numbers is unique up to permutation of the factors.
Here we require a stronger convention since,  e.g. $(\omega+1) \cdot \omega = \omega \cdot \omega$.
To this purpose Jacobsthal imposed a  condition on the sequence of prime factors.

\begin{theorem}[\cite{Jac09}]
\label{prime-factorization}
Every ordinal $\alpha$ has a unique factorization of the form
$\alpha=\omega^{A_{1}}A_{2}$ where $A_{1}\geq 0$ and 
$$
\begin{array}{l}
\omega^{A_{1}}=(\omega^{\omega^{\xi_{1}}})^{n_{1}}\cdots (\omega^{\omega^{\xi_{r}}})^{n_{r}}\\
A_{2}=a_{0}(\omega^{\mu_{1}} +1)a_{1}(\omega^{\mu_{2}} +1) \cdots a_{n-1}(\omega^{\mu_{n}} +1) a_{n}
\end{array}
$$
for some ordinals $\xi_{1}>\xi_{2} > \cdots > \xi_{r}$,
for some integers $n\geq 0$, $a_{0}, a_{1}, \ldots, a_{n}>0$,  and some ordinals 
$\mu_{1}, \mu_{2}, \ldots, \mu_{n} \geq 1$.
We say that $A_{2}$ is the  \emph{maximal successor right factor}
of the ordinal $\alpha$.
\end{theorem}
Observe that the condition on the exponents of limit primes is necessary
if one wants to guarantee unicity: for instance, $\omega$ and
$\omega^{\omega}$ are primes and 
$\omega^{\omega}= \omega \times \omega^{\omega}$. 
The prime factorization of the product of two prime factorizations follows the 
rule
\begin{equation}
\label{eq:product-of-factorizations}
\omega^{A_{1}} A_{2} \times \omega^{B_{1}} B_{2} 
= \left\{
\begin{array}{ll}
\omega^{A_{1}} A_{2}  B_{2} & \text{ if } B_{1}=0\\
\omega^{A_{1} + \partial(A_{2})+ B_{1}} B_{2} & \text{ if } B_{1} \ne 0 
\end{array}
\right.
\end{equation}

\subsection{Logic}


Let us specify our logical conventions and notations. We work within first-order predicate calculus
with equality.
Given a signature $\+L$, an $\+L$-\emph{structure} is denoted as $\langle D; \+L \rangle$ where $D$ is the \emph{base set} of the structure
and $ \+L$ is a set of interpreted relation and function symbols. We will actually 
confuse formal symbols and their interpretations.

We recall that a predicate is 
 $\Sigma_{n}$ (resp.  $\Pi_{n}$) if it is defined by a  formula 
 that begins with some existential (resp. universal) quantifiers and alternates $n-1$  
 times between blocks of existential and universal quantifiers.
  A \emph{fragment} of a theory
 is the set of sentences  of bounded complexity. E. g., the existential 
 fragment is the set of sentences of complexity $\Sigma_{1}$. It is a general
 concern when dealing with an undecidable theory and a motiviation 
 for our work to investigate 
its fragment with lowest complexity which is undecidable.

 We will, when possible or when relevant, consider not only the number of alternations
 of blocks, but also the sizes of the blocks by specifying the number
 of variables in each block. For example, if $\phi(x,y, z,t,u, \ldots)$ is a quantifier-free formula then  $\forall x,y\ \exists z,t,u\ \phi(x,y, z,t,u, \ldots)$
 has complexity $\Pi_{2}$ but we may say more precisely 
  that its complexity is $\forall  \forall \exists \exists \exists$, also written 
$\forall^{2} \exists^{3}$. When speaking of a set of formulas, not single formulas, 
we might say that their complexity is for example $\forall^{*} \exists^{3}$
meaning that for each formula there exists an integer $n$ such that its complexity
is $\forall^{n} \exists^{3}$.

We assume the reader has some familiarity with 
 computing the logical complexity. When evaluating  
 the complexity we will skip some steps to keep the computation
 readable. 
The next definition is crucial.

\begin{definition}
\label{de:definable}
Given an ${\cal L}-$structure $M= \langle {D; {\cal L}} \rangle$, 
an  $n-$ary relation $R$ over $D$
 is \emph{elementary definable} (shortly: \emph{definable}) in $M$ if there exists a ${\cal L}-$formula $\varphi$ with $n$ free variables such that $R=\{(a_1,\dots,a_n) :  M \models \varphi(a_1,\dots,a_n)\}$. Given a syntactic fragment $F$, we say that $R$ is \emph{$F-$definable} if $R$ is definable by a formula which belongs to $F$.
\end{definition}

\begin{example}
\label{ex:0-1}
In the structure  \str{\omega^{\omega^\lambda} ; \times } where $\lambda$ is an ordinal greater than $0$,
the formula $x\times x=x$
defines the set $\{0,1\}$ and the formula
 $(x\times x=x) \wedge \exists y\ (x\times y \not=x)$ defines the singleton $\{1\}$.
\end{example}

\section{The multiplicative monoid of ordinals}
\label{s:ordinals-as-a-monoid}

 We shall consider logical structures with domain an ordinal $\alpha$, which is  identified with the set of ordinals $\beta<\alpha$, and predicates and functions whose interpretation correspond to restrictions to $\alpha$ of relations and functions defined on the class of ordinals, such as the function $\times$. For simplicity we will use a single symbol for each restriction of a relation, e.g. we simply write  $\langle \alpha; \times  \rangle$. 
   Note that  \cite{bes02} dealt with $\times$ as a ternary relational symbol, which allows 
  {one} to consider any ordinal $\alpha$ as the base set of the structure. In this paper we consider $\times$ as a function symbol, which imposes that the base set of the structure is closed under multiplication, which holds if and only if it is an ordinal of the form $\omega^{\omega^\lambda}$ for some ordinal $\lambda$, as can be readily verified 
 from the basic definitions of the operations on ordinals.  

\medskip

\subsection{Elementary predicates}
\label{subsec:def-pred}

We show how the four constants $1, \omega, \omega+1$ and $\omega^{2}+1$
can be expressed in \str{\omega^{\omega^{\lambda}};\times } for every ordinal $\lambda>0$.
The idea is that  having these constants as free may lower the syntactical complexity of the 
predicates.

\begin{proposition}\label{prop:elemdef}
 For every ordinal $\lambda>0$, the following predicates are definable in  \str{\omega^{\omega^{\lambda}};\times} 
\begin{enumerate}
\item $\texttt{Zero}(x)=\{0\}$ is $\exists$-and $\forall$-definable

\item $\texttt{One}(x)=\{1\}$   is $\exists$-and $\forall$-definable

\item  $\Prime(x)=\{\alpha \mid \alpha \text{ is a prime}\}$ is  $\forall \forall$-definable:  %
 
 \item    $\LimPrime(x)=\{\omega^{\omega^{\xi}} \mid  \xi \text{ is an ordinal} \}$ is $\exists \forall$-
 and $\forall\exists$-definable.

\item  $\texttt{Omega}(x)=\{\omega\}$ is $\exists   \forall \forall$-definable.

\item $\texttt{OmegaPlusOne}(x)=\{\omega+1\}$ is $\exists \exists \forall \forall $-definable.

\item $\texttt{OmegaSquarePlusOne}(x)=\{\omega^2+1\}$ is   $\exists \exists \forall \forall $-definable.

\end{enumerate}
\end{proposition}

\begin{proof} We prove successively these assertions.

\begin{enumerate}

\item
 $\texttt{Zero}(x)$ can be defined by the formula  $ 
\forall y\ (x\times y=x)$ and by 
$\exists y\ (y\times y \ne y \wedge x \times y =x)
$

\item
$\texttt{One}(x)$ is definable by 
$ 
 \forall y\ (x\times y=y)$ and  by
 $x\times x=x  \wedge \exists   y\ (x \times y\not=x).
$

\item $\Prime(x)$ is definable by  $\forall  y,z\ A(x,y,z)$ where $A(x,y,z)$ is the formula 
$$x\times x\not =x \wedge    ( y\times y \not=y \wedge x=z\times y 
 \Rightarrow y=x).$$

 Indeed, it suffices to say that $x$ has exactly two right-hand side 
 divisors, namely $1$ and $x$.

\item   
  $
\LimPrime(x) $ is definable by $ \forall  y,z\ \exists t \ C(x,y,z,t)  $ and $  \exists t\  \forall  y,z \ C(x,y,z,t)\\
 $
where $C(x,y,z,t) = A(x,y,z) \wedge B(x,t)$ with 
 $B(x,t) = t\times t\not=t  \wedge t \times x=x$.

 It suffices to say that $x$ is a prime and that $t \times x=x$
for some $t\not=0,1$. Indeed if $x$ is a limit prime then e.g. $2x=x$. 
On the other hand if $x$ is a successor prime, then by unicity of 
factorization the equation $yx=x$ implies $y=1$.

\item  $\texttt{Omega}(x)$ is definable by $\exists t  \forall  y,z\ E(x,y,z,t) $\\
where $E(x,y,z,t) =C(x,y,z,t) \wedge  D(x,y,z)$
and 
$D(x,y,z)$ is  the formula $$y \times x=x \wedge z \times x=x \Rightarrow (y\times z= z \times y)
$$

 It suffices to say that $\omega$ is a limit prime and that all its left divisors 
 are integers,  thus ordinals that commute pairwise. Indeed, a limit
 prime different from $\omega$ is of the form $\omega^{\omega^{\alpha}}$
 with $\alpha>0$, but then $2\times \omega^{\omega^{\alpha}}= \omega \omega^{\omega^{\alpha}}= \omega^{\omega^{\alpha}}$
 and $2\times \omega= \omega\not= \omega\times 2$.  

\item   $\texttt{OmegaPlusOne}(x)$ can be defined by 
$$
\begin{array}{l}
  \Prime(x) \wedge  \exists u\ (\texttt{Omega}(u) \wedge x\times u= u \times u   \wedge x\times u  \ne u \times x)\\
 = \forall  y \forall z\ A(x,y,z)\wedge  \exists u ( \exists t  \forall  v,w\ E(u,v,w,t)\wedge x\times u= u \times u \ne u \times x)  \\
 = \exists u \exists t   \forall  y \forall z\  (A(x,y,z)\wedge   (E(u,y,z,t)\wedge x\times u= u \times u \ne u \times x))  \\
= \exists u \exists t   \forall  y \forall z\  F(x,y,z,u,t)  \\
\end{array}
$$

Indeed, it suffices to say that $\omega+1$ is the only prime $p\not=\omega$ which satisfies 
$p\times \omega=\omega^{2}$. Indeed, if $p$ is finite then $p\times \omega=\omega$.
If it is a nonfinite successor $\omega^{\alpha}+1$ then
$p\times \omega=\omega^{\alpha+1}=\omega^{2}$ implies $\alpha=1$. If
it is of the form $\omega^{\omega^{\alpha}}$ with $\alpha>0$ then
$\omega^{\omega^{\alpha}}\times \omega = \omega^{\omega^{\alpha}+1}\not=\omega^{2}$

\item 
$\texttt{OmegaSquarePlusOne}(x)$. Similarly to the above case 
we observe that the ordinal $\omega^{2}+1  $

is the only prime $p \ne \omega^2$ such that $p\times \omega= \omega^{3}$.
This leads to a formula of the form $
\exists u \exists t   \forall  y \forall z\  G(x,y,z,u,t)%
$
for some suitable quantifier-free formula $G$
(replacing the formula $F$ of the above case).

\end{enumerate}
\end{proof}

\section{Interpreting the existential fragment of
$\langle \N; +, |\rangle$ in the existential fragment of
$\langle \omega^{\omega^\lambda}; \times, \omega, \omega +1 , \omega^{2} +1\rangle$}
\label{s:interpretation-in-N-divibililty}

Let $\lambda>1$ be an ordinal. The full theory of $\langle \omega^{\omega^\lambda}; \times\rangle$ is undecidable, cf. \cite[Thm. 11]{bes02}.
Here we show that the existential fragment of 
$\langle \N; +, |\rangle$, cf. \cite{lipshitz}, can be interpreted 
in the existential fragment of the structure $\Omega= \langle \omega^{\omega^\lambda}; \times, 1, \omega, \omega +1 , \omega^{2} +1\rangle$ (i.e. the expansion of $\langle \omega^{\omega^\lambda}; \times \rangle$ with constants $1, \omega, \omega +1 , \omega^{2} +1$). 
 This gives a lower bound on the complexity of the latter fragment.

The idea is to take advantage of the fact that the products of the elements 
$\omega +1$ and $\omega^{2} +1$ define a  free monoid to which 
the elementary combinatorial 
property of Lemma \ref{le:commuting-successor-ordinals} applies.

\begin{lemma}
\label{le:indispensable-for-multiples}
For every integer $n \geq 0$,
the non null solutions of the equation 
$$z(\omega +1)^{n}(\omega^2 +1)=(\omega +1)^{n}(\omega^2 +1)z$$
in $ \omega^{\omega^\lambda}$ are of the form $((\omega +1)^{n}(\omega^2 +1))^{r}$ for arbitrary integers $r$.
\end{lemma}

\begin{proof}
By Theorem \ref{prime-factorization},  $z$ can be written as $z=\omega^{z_1}z_2$ where $z_2$ is a successor ordinal. Applying again Theorem \ref{prime-factorization} to each member of the equation $z(\omega +1)^{n}(\omega^2 +1)=(\omega +1)^{n}(\omega^2 +1)z$ implies that $z_1=0$, i.e. that $z$ is a successor ordinal. By Lemma \ref{le:commuting-successor-ordinals} this implies that there exist $\alpha$ and $j,r<\omega$ such that 
$(\omega +1)^{n}(\omega^2 +1)=\alpha^{j}$ and $z=\alpha^{r}$. The former condition implies  $j=1$,
 thus  $z=((\omega +1)^{n}(\omega^2 +1))^{r}$.  
\end{proof}

The interpretation of $\langle \N; +, |,0,1\rangle$ in $\Omega$ consists {of} identifying any integer $i$ with the ordinal $(\omega +1)^{i}$.

\begin{lemma}\label{lem:dom}
	\label{le:Div}
	The set $\Dom=\{(\omega+1)^{i} \mid 0 \leq i  <\omega \}$ 
	is definable in $\Omega$ by the quantifier-free formula $$\theta(x): \quad x(\omega+1)= (\omega+1)x \wedge x (\omega+1) \ne x.$$
\end{lemma}

\begin{proof}
	If $x \in \Dom$ then it is clear that $\theta(x)$ holds. For the converse assume that $\theta(x)$ holds. The condition $x (\omega+1) \ne x$ implies $x \ne 0$. By Theorem \ref{prime-factorization},  $x$ has a unique factorization of the form $x=\omega^{A_1}A_2$ where $A_2$ is a successor ordinal. The equality $ x(\omega+1)= (\omega+1)x$ implies $\omega^{A_1}A_2(\omega+1)=(\omega+1)w^{A_1}A_2$. Using again  Theorem \ref{prime-factorization},  we can deduce that $A_1=0$, i.e. that  $x$ is a successor ordinal. Therefore we can apply Lemma \ref{le:commuting-successor-ordinals} to the equation $x(\omega+1)= (\omega+1)x$, which yields $x \in \Dom$.
\end{proof}

\begin{proposition}\label{prop:interpdiv}
For every existential closed formula $\exists x_1,\dots,x_n \phi(x_{1}, \ldots, x_{n})$ of $\langle \N; +, |,0,1\rangle$ 
there exists an  existential closed formula $\exists x_1,\dots,x_n\Phi(x_{1}, \ldots, x_{n})$ of $\Omega$  such that the following properties hold:
\begin{enumerate}
\item For all 
$i_{1}, \ldots, i_{n}\in \N$:
$$
\begin{array}{c}
\langle \N; +, |,0,1\rangle \models \phi(i_{1}, \ldots, i_{n})\\
\Leftrightarrow \\
\Omega  \models \Phi( (\omega+1)^{i_{1}}, \ldots, (\omega+1)^{i_{n}})
\end{array}
$$
\item For all ordinals $\alpha_{1}, \ldots, \alpha_{n}$ the condition 
$\Omega \models \Phi(\alpha_{1}, \ldots, \alpha_{n})$ implies that all $\alpha_i$'s are powers of $\omega+1$ {with integer exponent}. 

\end{enumerate} 
\end{proposition}

The proof of Proposition \ref{prop:interpdiv} relies on the possibility to interpret divisibility of integers in $\Omega$.

\begin{lemma}
\label{le:Div}
The predicate $\Div(x,y) =\{((\omega+1)^{i}, (\omega+1)^{j})\mid i  \text{ divides }   j\}$ 
is definable in $\Omega$ by a formula of the 
form  $\exists z,t\  E(x,y,z,t)$  
where $E$ is quantifier-free.
\end{lemma}

\begin{proof}

Consider the following predicates 
\begin{align*}
A(x,y,t) &:&\theta(x) \wedge \theta(y) \wedge x=t(\omega+1)^{2},\\
B(z,t) &:&zt(\omega^2+1)=t(\omega^2+1)z,\\
C(y,z) &:&y\omega =z\omega.
\end{align*}

where $\theta$ is the formula of Lemma \ref{lem:dom}.
\begin{itemize}

\item $A(x,y,t)$ is equivalent to the existence of $i, j \geq 0$ such that
$x=(\omega+1)^{i+2}$,  $t=(\omega+1)^{i}$ and $y=(\omega+1)^{j}$. 

\item  Via Lemma \ref{le:indispensable-for-multiples}, $B(z,t)$ is equivalent to the existence of an integer $r$
such that $z=((\omega+1)^{i}(\omega^2+1))^{r}$.

\item $C(y,z)$ implies  that $j=(i+2) r$.

\end{itemize}

Set $D(x,y,z,t)= A(x,y,t) \wedge B(z,t) \wedge C(y,z)$.
For all ordinals  $\alpha, \beta$ we have
$$
\Omega \models \exists z,t\ D(\alpha, \beta,z,t)
  \text{ iff }   \alpha=(\omega+1)^{i+2}, \beta= (\omega+1)^{j}\text{ where }  
(i+2) \text{  divides }  j.
$$
It remains the case $i=1$ which divides whatever value of $j$  
and the case $i=j=0$. 
Then $\Div(x,y)$ can be defined in $\Omega$ by the formula 

\begin{multline}\label{eq:expressing-divides}
(x=\omega +1 \wedge y(\omega +1) =(\omega +1)y)  \vee(x=y=1)  
 \vee  \exists z,t\   D(x,y,z,t)\\
  \equiv \exists z,t\  E(x,y,z,t)
\end{multline}

where $E(x,y,z,t): \ (x=\omega +1)  \vee(x=y=1)  \vee  ( D(x,y,z,t))$
is quantifier-free.
\end{proof}

Now we can prove Proposition  \ref{prop:interpdiv}.

\begin{proof}
Any existential closed formula in $\langle \N; +, |,0,1\rangle$
is equivalent to a disjunction of formulas of the form
$$
\displaystyle \exists x_{1}, \ldots, x_{n}  \bigwedge_{k} \ell_{k} (x_{1}, \ldots, x_{n})| r_{k} (x_{1}, \ldots, x_{n})
$$
where $\ell_{k} (x_{1}, \ldots, x_{n})$ and $r_{k} (x_{1}, \ldots, x_{n})$ 
are linear expressions in the variables $x_{1}, \ldots, x_{n}$. 
Thus we can assume w.l.o.g. that $\phi(x_{1}, \ldots, x_{n})$ has the form $\bigwedge_{k} \ell_{k} (x_{1}, \ldots, x_{n})| r_{k} (x_{1}, \ldots, x_{n})$. 

We recall that the interpretation goes by  identifying the integer $i$ with the ordinal $(\omega +1)^{i}$. Any linear combination $i_{1} + \cdots + i_{p}+a$ is thus 
identified with the product 
$$
(\omega +1)^{i_{1}} \cdots (\omega +1)^{i_{p}}(\omega +1)^{a}
= (\omega +1)^{i_{1} +\cdots + i_{p}+a}.
$$

The formula $\Phi(x_{1}, \ldots, x_{n})$ is defined as $\bigwedge_{k} \gamma_k(x_{1}, \ldots, x_{n})$ where each $\gamma_k$ is a quantifier-free formula which translates in $\Omega$ the formula $\ell_{k} (x_{1}, \ldots, x_{n})| r_{k} (x_{1}, \ldots, x_{n})$.
The latter formula has the form 
$$
(x_{i_1} + \cdots + x_{i_p} + a)  | (x_{j_1} + \cdots + x_{j_q} +b)
$$
where $1 \leq i_1,\dots,i_p,j_1,\dots,j_q \leq n$ and $a,b$ are integers. 
We can define $\gamma_k$ as  
$$
\begin{array}{ll}
\displaystyle \exists x,y,z,t \   &( E(x,y,z,t) \wedge \displaystyle \bigwedge^{p}_{k=1} \theta(x_{i_{k}})
\wedge  \bigwedge^{q}_{\ell=1} \theta(x_{j_{\ell}})
\\
& \wedge\  x= x_{i_1} \cdots x_{i_p} (\omega +1)^{a} \wedge  y= x_{j_1} \cdots x_{j_q} (\omega +1)^{b})
\end{array}
$$
where $\theta$ is the formula of Lemma \ref{lem:dom}.
\end{proof}

Proposition  \ref{prop:interpdiv} gives  a lower bound for the satisfiability of the 
existential fragment of the theory $\Omega$.

\begin{corollary}
The satisfiability of the 
existential fragment of the theory 
$\Omega$
is at least as hard as that of the satisfiability of the 
existential fragment of the theory 
$\langle \N; +, |,0,1\rangle$ which is in \text{NEXPTIME} (see \emph{\cite{lechner}}).
\end{corollary}

\section{Undecidability of the $\Sigma_2-$fragment of 
$\langle \omega^{\omega^\lambda}; \times\rangle$}
\label{s:undecidability}

Our objective is to interpret the Diophantine fragment of the nonnegative integers
 in the simplest fragment of  $\Omega$  as possible.
This is achieved in two steps. First by expressing in $\Omega$ {the
least common multiple of two integers, abbreviated  \lcm,} in terms of divisibility, and then 
{by expressing} the multiplication in terms of \lcm. 

As a corollary we obtain undecidability of  the $\Sigma_2-$fragment of $\Omega$ and 
finally  of 
$\langle \omega^{\omega^\lambda}; \times \rangle$ by removing the constants.

\medskip By using the same construction as in the proof of Proposition \ref{prop:interpdiv} (where we identify the integer $i$ with the ordinal $(\omega +1)^{i}$) we are led to interpret the  function $\lcm$ of integers 
in $\Omega$ and to  investigate the complexity of this predicate.
We use the  $\Sigma_{1}$-predicate  $\Div(x,y)$ introduced in Lemma \ref{le:Div}
{and set}
$$
\LCM(x,y,z)=\{((\omega+1)^{i}, (\omega+1)^{j},  (\omega+1)^{k} ) \mid k=\lcm(i,j) \}.
$$

\begin{lemma}\label{le:interp-lcm}
The predicate  $\LCM(x,y,z)$ is definable in $\Omega$ by a  $\exists^{4} \forall^{6}-$formula.

\end{lemma}
\begin{proof}
The predicate $\LCM(x,y,z)$ is definable by the formula 
$$
\begin{array}{l}
\Div(x,z) \wedge \Div(y,z) \wedge \forall u ((\Div(x,u) \wedge \Div(y,u) \wedge u\omega \ne z \omega)\\
 \rightarrow \forall t (tu\omega \ne z \omega)).
\end{array}
$$
Indeed, this formula can be interpreted as saying that $x=(\omega+1)^{i}, y=(\omega+1)^{j}$
and $z=(\omega+1)^{k}$ where $i$ and $j$ divide $k$ and  all integers $\ell\not=k$ 
divisible by $i$ and $j$ are greater than $k$. Substitute expression \ref{eq:expressing-divides}
for each occurrence of $\Div$.
The following  
sequence of equivalences leads to a complexity in $\exists^{4}\ \forall^{6} $, namely
$$
\begin{array}{l}
\LCM(x,y,z)\equiv \\
\exists u_{1}, u_{2}\ E(x,z,u_{1},u_{2}) \wedge \exists u_{3},u_{4}\ E(y,z,u_{3},u_{4})\\
 \wedge \forall v ((\exists v_{1},v_{2}\ E(x,v,v_{1},v_{2}) \wedge \exists v_{3},v_{4}\  E(y,v,v_{3},v_{4}) \wedge v\omega \ne z \omega)\\
 \rightarrow \forall v_{5} (v_{5} v \omega \ne z \omega))\\
 \equiv  \exists u_{1}, u_{2}\ E(x,z,u_{1},u_{2}) \wedge \exists u_{3},u_{4}\ E(y,z,u_{3},u_{4})\\
 \wedge \forall v ( \forall v_{1},v_{2}\ (\neg E(x,v,v_{1},v_{2})) \vee \forall v_{3},v_{4}\  (\neg E(y,v,v_{3},v_{4})) \vee (v\omega = z \omega)\\
 \vee \forall v_{5} (v_{5}v\omega \ne z \omega))\\
 \equiv  \exists u_{1}, u_{2}, u_{3},u_{4}\ (E(x,z,u_{1},u_{2}) \wedge  \ E(y,z,u_{3},u_{4}))\\
 \wedge \forall v, v_{1}, v_{2}, v_{3}, v_{4}, v_{5} \ (\neg E(x,v,v_{1},v_{2}) \vee   \neg E(y,v,v_{3},v_{4}) \vee (v\omega = z \omega)\\
 \vee (v_{5}v\omega \ne z \omega))\\
 = \exists^{4} \overrightarrow{u}\ \forall^{6} \overrightarrow{v}\  F(x,y,z, \overrightarrow{u}, \overrightarrow{v})
\end{array}
$$
where $F$ is quantifier-free.

\end{proof}
Multiplication is expressible
in terms of the least common multiple. Indeed, a simple computation 
shows 
\begin{equation}
\label{eq:times-as-lcm}
  x\times y = \frac{1}{2}(\lcm(x+y,x+y+1)- \lcm(x,x+1)- \lcm(y,y+1)).
\end{equation}

We are ready to show that the multiplication can be interpreted in $\Omega$
at a low cost. We set 
$$
\MULT(x,y,z)=\{((\omega+1)^{i}, (\omega+1)^{j},  (\omega+1)^{k} ) \mid i \times j = k \}.
$$

\begin{lemma}
\label{le:product}
The predicate $\MULT(x,y,z)$ is definable in $\Omega$ by a $\exists^{15} \forall^{6}-$formula.
\end{lemma}

\begin{proof}
In view of expression \ref{eq:times-as-lcm} we can define $\MULT(x,y,z)$ as follows:
$$
\begin{array}{l}
\exists r_{1} r_{2} r_{3}\ ( z\times z\times  r_{1}\times r_{2} = r_{3} \wedge \LCM(x\times  y,  x\times  y  (\omega+1), r_{3})  \wedge 
 \LCM(x, x  (\omega+1), r_{1})\\
\phantom{xxxxxxx}   \wedge \LCM(y, y(\omega+1), r_{2})).
  \end{array}
$$
Using the expression given in the proof of Lemma \ref{le:interp-lcm} we get
$$
\begin{array}{l}
\exists r_{1} r_{2} r_{3}\  ((z\times z\times  r_{1}\times r_{2} = r_{3}) \wedge (\exists^{4} \overrightarrow{u_{3}}\ \forall^{6} \overrightarrow{v_{3}}
F(x\times  y,  x\times  y \times  (\omega+1), r_{3}, \overrightarrow{u_{3}}, \overrightarrow{v_{3}}))\\
  \wedge (\exists^{4} \overrightarrow{u_{1}}\ \forall^{6} \overrightarrow{v_{1}}
F(x, x \times  (\omega+1), r_{1}, \overrightarrow{u_{1}}, \overrightarrow{v_{1}}))
  \wedge (\exists^{4} \overrightarrow{u_{2}}\ \forall^{6} \overrightarrow{v_{2}}
F(y,y \times  (\omega+1), r_{2}, \overrightarrow{u_{2}}, \overrightarrow{v_{2}}))\\
\equiv   
\exists r_{1} r_{2} r_{3}\  \exists^{4} \overrightarrow{u_{1}}\
 \exists^{4} \overrightarrow{u_{2}}\ \exists^{4} \overrightarrow{u_{3}}\
  \forall^{6} \overrightarrow{v}\\
\quad     ((z\times z\times  r_{1}\times r_{2} = r_{3})  \wedge
F(x\times  y,  x\times  y  (\omega+1), r_{3}, \overrightarrow{u_{3}}, \overrightarrow{v})\\
\quad  \wedge 
F(x, x \times  (\omega+1), r_{1}, \overrightarrow{u_{1}}, \overrightarrow{v})
  \wedge 
F(y,y \times  (\omega+1), r_{2}, \overrightarrow{u_{2}}, \overrightarrow{v}))\end{array}
$$
\end{proof}

We wish to express in $\Omega$ all terms of $\langle \N;+, \times, 0,1\rangle$. We start with terms $W$ which are products of variables. 

\begin{definition}
Let  $W(x_{1}, \ldots, x_{k})$ be a term of $\langle \N;+, \times, 0,1\rangle$ which consists {of} a product of occurrences 
among of the
variables $x_{1}, \ldots, x_{k}$. We set 
$$
\TERM_W(y,x_1,\dots,x_k)=\{((\omega+1)^{i}, (\omega+1)^{j_1},\dots,  (\omega+1)^{j_k} ) \mid i=W(j_1,\dots,j_k) \}.
$$
\end{definition}

\begin{lemma}\label{le:interp-terms}
For each term $W(x_{1}, \ldots, x_{k})$  which consists {of} a product of occurrences 
among of the
variables $x_{1}, \ldots, x_{k}$, the predicate $\TERM_W(y,x_1,\dots,x_k)$ is definable in $\Omega$ by a formula of the form
$$
\exists^{m} \overrightarrow{u}\ \forall^{6} \overrightarrow{v}\  
G(y, x_{1}, \ldots, x_{k} , \overrightarrow{u}, \overrightarrow{v})
$$
where $m$ depends on the term, and $G$ is quantifier-free.

\end{lemma}

\begin{proof}
We proceed by induction on the length of the product.
If the product is reduced to a single variable then by using dummy variables 
it is always possible to put the equality $y=x$ in the appropriate form. 
For the induction step, consider two terms  $W_{1}(x_{1}, \ldots, x_{k})$ and $W_{2}(x_{1}, \ldots, x_{k})$, and assume that $\TERM_{W_1}(y,x_1,\dots,x_k)$ and $\TERM_{W_2}(y,x_1,\dots,x_k)$ are definable in $\Omega$ by the formulas
$$
 \exists^{m_{1}} \overrightarrow{u_{1}}\ \forall^{6} \overrightarrow{v_{1}}\  
G_{1}(y_{1}, x_{1}, \ldots, x_{k} , \overrightarrow{u_{1}}, \overrightarrow{v_{1}})$$ and
$$\exists^{m_{2}} \overrightarrow{u_{2}}\ \forall^{6} \overrightarrow{v_{2}}\  
G_{2}(y_{2}, x_{1}, \ldots, x_{k} , \overrightarrow{{u_{2}}}, \overrightarrow{v_{2}}),
$$
respectively, where the bound variables which are the components of
$u_{1}$ and $u_{2}$ are pairwise different.
Then $\TERM_{W_1 \times W_2}(y,x_1,\dots,x_k)$ can be defined by the formula: 
\begin{multline}\label{eq:prodW}
	\exists y_{1}, y_{2}\   ((y= y_{1} y_{2}) \wedge
	\exists^{m_{1}} \overrightarrow{u_{1}}\ \forall^{6} \overrightarrow{v_{1}}\  
	G_{1}(y_{1}, x_{1}, \ldots, x_{k} , \overrightarrow{u_{1}}, \overrightarrow{v_{1}})\\
	\wedge \exists^{m_{2}} \overrightarrow{u_{2}}\ \forall^{6} \overrightarrow{v_{2}}\  
	G_{2}(y_{2}, x_{1}, \ldots, x_{k} , \overrightarrow{u_{2}}, \overrightarrow{v_{2}})).
\end{multline}	
	
By applying the rule $\exists x\  f(x, \ldots) \wedge \exists y\  g(y, \ldots)\equiv
\exists x\  f(x, \ldots) \wedge   g(x, \ldots)$  provided $x$ and $y$ are different, $y$ does not occur 
in $f$ and $x$ in $g$, and the rule $\forall x\  f(x, \ldots) \wedge \forall y\  g(y, \ldots)\equiv
\forall x\  (f(x, \ldots) \wedge   g(x, \ldots))$), the formula \ref{eq:prodW}  is equivalent to 
\begin{multline}
\exists y_{1}, y_{2}\ \exists^{m_{1}} \overrightarrow{u_{1}}\ \exists^{m_{2}} \overrightarrow{u_{2}}\
\forall^{6} \overrightarrow{v_{1}} ((y= y_{1} y_{2}) \wedge\\
G_{1}(y_{1}, x_{1}, \ldots, x_{k} , \overrightarrow{u_{1}}, \overrightarrow{v_{1}})  \wedge G_{2}(y_{2}, x_{1}, \ldots, x_{k} , \overrightarrow{u_{2}}, \overrightarrow{v_{1}})).
\end{multline}
\end{proof}

The ultimate goal is to apply Matjasevich undecidability result for the Diophantine fragment
of the integers with addition and multiplication \cite{Mat70} to show that the 
$\exists^*\ \forall^{6}$-fragment of $\Omega$ is undecidable.

\begin{definition}
	Let $E(x_{1}, \ldots, x_{p})$ be an atomic formula of $\langle \N; +, \times, 0,1\rangle$ (i.e. a Diophantine equation in $\N$). We set 
$$
\EQ_E(x_1,\dots,x_p)=\{( (\omega+1)^{j_1},\dots,  (\omega+1)^{j_p} ) \mid \langle \N; +, \times, 0,1\rangle \models E(j_1,\dots,j_p) \}.
$$

\end{definition}

\begin{theorem}
For each atomic formula $E(x_{1}, \ldots, x_{p})$ of  $\langle \N; +, \times,  0,1\rangle$, the predicate $\EQ_E(x_1,\dots,x_p)$ is definable in $\Omega$ by a
$\exists^*\forall^{6}$-formula.
\end{theorem}
\begin{proof} The equation $E(x_1,\dots,x_p)$ can be written as
$$
W_{1}(x_{1}, \ldots, x_{p}) + \cdots + W_{n}(x_{1}, \ldots, x_{p})=W_{n+1}(x_{1}, \ldots, x_{p}) + \cdots + W_{n+m}(x_{1}, \ldots, x_{p})
$$ 
where each $W_{i}$ is a product of occurrences of variables among $x_{1}, \ldots, x_{p}$.
By Lemma \ref{le:interp-terms} for $i=1, \ldots,  n+m$ 
there exists a formula 
$$
\Phi_{i}\equiv \exists^{n_{i}} \overrightarrow{u_{i}}\ \forall^{6} \overrightarrow{v_{i}}\phi_{i}(y_{i}, x_{1}, \ldots, x_{p}, \overrightarrow{u}, \overrightarrow{v})
$$ 
which defines $\TERM_{W_i}(y_{i}, x_{1}, \ldots, x_{p})$ in $\Omega$.
Set $\Phi(x_{1}, \ldots, x_{p})$ as the formula
\begin{multline}
\exists y_{1},  \cdots, y_{n +m}  
\displaystyle  \big(\bigwedge^{n+m}_{i}  \exists \overrightarrow{u_{i}}\ 
\forall^{6} \overrightarrow{v_{i}}\ 
\Phi_{i}(y_{i}, x_{1}, \ldots, x_{p}, \overrightarrow{u_{i}}, \overrightarrow{v_{i}})\big)\\
\wedge (y_{1} \times \cdots \times y_{n}= y_{n+1} \times \cdots \times y_{n+m}).
\end{multline}

It is clear that $\Phi(x_{1}, \ldots, x_{p})$ defines $\EQ_E(x_1,\dots,x_p)$ in $\Omega$.
By routine rewriting on quantifiers and renaming of bound variables,  the above formula is equivalent to
$$
\begin{array}{l}
 \exists y_{1},  \cdots, y_{n +m}  \exists \overrightarrow{u_{1}}\ \ldots \exists \overrightarrow{u_{n+m}}\ \forall \overrightarrow{v}\\
\displaystyle  \big(\bigwedge^{n+m}_{i}  \Phi_{i}(y_{i}, x_{1}, \ldots, x_{p}, \overrightarrow{u_{i}}, \overrightarrow{v})\big)
 \wedge (y_{1} \times \cdots \times y_{n}= y_{n+1} \times \cdots \times y_{n+m})
\end{array}
$$
which is a formula of complexity $\exists^*\ \forall^{6}$.
\end{proof}

\begin{corollary}
The $\exists^*\forall^{6}-$fragment of 
$\langle \omega^{\omega^\lambda}; \times \rangle$ is undecidable for every ordinal $\lambda \geq 1$.
\end{corollary}

\begin{proof}
We show how we can remove the constants $1$, $\omega$, $(\omega+1)$
and $(\omega^{2}+1)$.
By Proposition \ref{prop:elemdef}, there exist quantifier-free predicates $\alpha, \beta, \gamma, \delta$
such that 
$$
\begin{array}{rll}
\{1\} &  \text{is definable in $\langle \omega^{\omega^\lambda}; \times \rangle$ by} &  \exists y\ \alpha(x,y)\\
\{\omega\} &  \text{is definable by}  &  \exists y \forall z,t\  \beta(x,y,z,t)\\
\{\omega+1 \} & \text{is definable by}  &   \exists y_1 \exists y_2 \forall z,t\  \gamma(x,y,z,t)\\
\{\omega^{2}+1\} & \text{is definable by}  &   \exists y_1 \exists y_2 \forall z,t\  \delta(x,y,z,t)\\
\end{array}
$$
Thus every formula
$$
\exists \overrightarrow{u}\ \forall \overrightarrow{v}\phi(x_{1}, \ldots, x_{p}, \overrightarrow{u}, \overrightarrow{v})
$$
over $\Omega$,   where $\overrightarrow{v}=(v_1,\dots,v_6)$ and $\phi$ is quantifier-free, is equivalent in 
$\langle \omega^{\omega^\lambda}; \times \rangle$ to the formula
$$
\begin{array}{l}
\exists a,b,c,d \ (  \exists y\ \alpha(a,y) \wedge  \exists y \forall z,t\  \beta(b,y,z,t)
\wedge \exists y_1 \exists y_2\forall z,t\  \gamma(c,y,z,t) \\
\quad
\wedge  \exists y_1 \exists y_2 \forall z,t\  \delta(d,y,z,t) \wedge 
\exists \overrightarrow{u}\ \forall \overrightarrow{v}\Phi(x_{1}, \ldots, x_{p}, \overrightarrow{u}, \overrightarrow{v}))
\end{array}
$$
where $\Phi$ is obtained from $\phi$ by substituting
$a$, $b$, $c$ and $d$ for every occurrence of the constants 
$1$, $\omega$, $\omega+1$ and $\omega^{2}+1$ and where 
$a,   b,  c,  d,  y, z, t$ are pairwise different variables 
which are also different from the components of $ \overrightarrow{u}$
and $\overrightarrow{v}$. 
By introducing four new variables $ z_{1},  z_{2},  z_{3},  z_{4}$
this formula is equivalent to
$$
\begin{array}{l}
\exists a,   b,  c,  d,  z_{1},  z_{2},  z_{3},  z_{4}, \overrightarrow{u}\\
  \quad ( \alpha(a,z_{1} ) \wedge    \forall z,t\  \beta(b,z_{2},z,t) \wedge   \forall z,t\  \gamma(c,z_{3},z,t) 
\wedge  \forall z,t\  \delta(d,z_{4},z,t)\\ 
\quad \wedge 
\forall \overrightarrow{v}\Phi(x_{1}, \ldots, x_{p}, \overrightarrow{u}, \overrightarrow{v}))
\end{array}
$$
which is equivalent to
$$
\begin{array}{l}
\exists a,   b,  c,  d,  z_{1},  z_{2},  z_{3},  z_{4}, \overrightarrow{u}\\
\quad \forall \overrightarrow{v}  \ ( \alpha(a,z_{1} ) \wedge \beta(b,z_{2},v_1,v_2) \wedge  \gamma(c,z_{3},v_1,v_2) 
\wedge \delta(d,z_{4},v_1,v_2)\\ 
\quad \wedge 
\Phi(x_{1}, \ldots, x_{p}, \overrightarrow{u}, \overrightarrow{v}))
\end{array}
$$
\end{proof}

\section{Final observation}
\label{s:final-observations}

Our results leave open the question of whether the existential fragment of $\Omega$ is decidable. By Section \ref{s:interpretation-in-N-divibililty}, a proof of decidability for 
the existential fragment of $\Omega$ would provide a new proof of decidability for the existential fragment of
$\langle \N; +, |\rangle$. We show that it would also provide a new proof of Makanin's result of decidability for word equations with constants \cite{mak77}. This suggests that a conceptually simple proof of decidability for the existential fragment of $\Omega$ may be difficult to obtain.

The relation between the problem we tackled and 
that of word equations with constants comes from the fact that the multiplicative monoid of the structure
$\Omega= $\str{\alpha; \times, 1, \omega, \omega+1, \omega^{2}+1} 
(where $\alpha=\omega^{\omega^{\lambda
	}}$ for some  $\lambda>0$) has an (infinitely generated)
	free submonoid, namely that generated by the infinite successor primes less than $\alpha$.

	Given a word equation
	with constants over a binary alphabet $\{a,b\}$ (which is no loss of generality
	because the unary case is trivial and the general case with more than one letter
	reduces to the binary case)

\begin{equation}
\label{eq:makanin}
L(x_{1}, \ldots, x_{n},a,b)= R(x_{1}, \ldots, x_{n},a,b),
\end{equation}
we define the following
conditions over $\Omega$
\begin{equation}
\label{eq:interpreting-makanin}
\begin{array}{l}
L'(x_{1}, \ldots, x_{n},(\omega+1), (\omega^{2}+1))= R'(x_{1}, \ldots, x_{n},(\omega+1), (\omega^{2}+1)) \wedge\\
(\omega+1) x_{1}\not= \omega x_{1} \wedge \ldots \wedge (\omega+1) x_{n}\not= \omega x_{n}
\end{array}
\end{equation}
where $L'$ (resp. $R'$) is obtained from $L$ (resp. $R$) by substituting 
$(\omega+1)$ and $(\omega^{2}+1)$ for $a$ and $b$.
Every solution $\theta(x_{i})= u_{i}\in \{a,b\}^*$ of \ref{eq:makanin}
yields a solution for \ref{eq:interpreting-makanin}, namely
$\theta'(x_{i})= f(u_{i})\in \alpha$ where $f$ substitutes 
$(\omega+1)$ and $(\omega^{2}+1)$ for every occurrence of $a$ and $b$, respectively.
Conversely, if $\theta'(x_{i})= \alpha_{i}\in \alpha$ is a solution 
for \ref{eq:interpreting-makanin}, then each inequality $(\omega+1) x_{i}\not= \omega x_{i}$
implies that $x_{i}$ is a successor. Now observe that the successors define a submonoid 
of $\alpha$. On this submonoid consider the mapping $g$ 
which maps every
$$x=a_{0}(\omega^{\mu_{1}} +1)a_{1}(\omega^{\mu_{2}} +1) \cdots a_{n-1}(\omega^{\mu_{n}} +1) a_{n}$$
where $\mu_i$'s are non null ordinals and the $a_i$'s are positive integers,
to $$g(x)=(\omega^{\nu_{1}} +1)(\omega^{\nu_{2}} +1) \cdots (\omega^{\nu_{n}} +1) 
$$ 
with $\nu_{i}=\min \{\mu_{i}, 2\}$. It is clearly a morphism and if $\theta'$ is
a solution of \ref{eq:interpreting-makanin}, so is $g\circ \theta'$. Composing with the morphism 
$h:\{(\omega  +1), (\omega^{2} +1) \}^*\rightarrow \{a,b\}$
defined by $h((\omega  +1))=a, h((\omega^{2}  +1))=b$, yields a solution 
of  \ref{eq:makanin} by setting $\theta(x_{i})= h(g(\theta'(x_{i}))$.

\end{document}